\documentclass[11pt]{article}
\newcommand{\marginppp}[1]{}
\pdfoutput=1
\usepackage{graphicx}
\usepackage{parskip}
\usepackage{fullpage}
\usepackage[dvipsnames]{xcolor}

\usepackage{bm}
\usepackage[charter]{mathdesign}
\usepackage{mathtools}
\usepackage{amsthm}
\usepackage{doi}
\usepackage[inline]{enumitem}
\setlist{noitemsep,leftmargin=*}
\hypersetup{colorlinks,linkcolor={red!50!black},citecolor={blue!50!black},
  urlcolor={blue!80!black}}
\newcommand{\Xomit}[1]{}
\newtheorem{theorem}{Theorem}
\newtheorem{lemma}{Lemma}
\title{Blazing a Trail via Matrix Multiplications: A Faster Algorithm
  for Non-shortest Induced Paths\footnote{A preliminary version~\cite{ChiuL22} appeared in \emph{Proceedings
    of the 39th International Symposium on Theoretical Aspects of
    Computer Science  (STACS)},  Marseille, March 15-18 2022.}}
\author{Yung-Chung Chiu\footnote{Department of Computer Science and
Information Engineering, National Taiwan University.}
\and
Hsueh-I Lu\footnote{Corresponding author.  Department of Computer
Science and Information Engineering, National Taiwan University.
Email: hil@csie.ntu.edu.tw.}}
\begin{document}
\maketitle
\begin{abstract}
For vertices $u$ and $v$ of an $n$-vertex graph $G$, a
\emph{$uv$-trail} of $G$ is an induced $uv$-path of $G$ that is not a
shortest $uv$-path of $G$.  Berger, Seymour, and Spirkl
[\emph{Discrete Mathematics} 2021] gave the previously only known
polynomial-time algorithm, running in $O(n^{18})$ time, to either
output a $uv$-trail of $G$ or ensure that $G$ admits no~$uv$-trail.
We reduce the complexity to the time required to perform a
poly-logarithmic number of multiplications of $n^2\times n^2$ Boolean
matrices, leading to a largely improved $O(n^{4.75})$-time algorithm.
\end{abstract}

\section{Introduction}
\label{section:section1}
Let $G$ be an $n$-vertex simple, finite, undirected, and unweighted graph.  
\marginppp{p2c1} Let $V(G)$ (respectively, $E(G)$) denote
\marginppp{p2c2} the vertex (respectively, edge) set of~$G$.  For any
subgraph~$H$ of $G$, let $G[H]$ be the subgraph of~$G$ induced by
$V(H)$.  A subgraph~$H$ of $G$ is \emph{induced} if~$G[H]=H$.  That
is, an induced subgraph of $G$ is a subgraph of $G$ that can be
obtained by deleting a set of vertices together with its incident
edges from $G$.  Various kinds of induced subgraphs are involved in
the deepest results of graph theory and graph algorithms.
\marginppp{p2c3} One of the most prominent examples concerns ``perfect
graphs''.  A graph $G$ is \emph{perfect} if every induced subgraph $H$
of $G$ has chromatic number equal to its clique number.  A graph is
\emph{odd} (respectively, \emph{even}) if it has an odd (respectively,
even) number of edges.  A \emph{hole} of $G$ is an induced cycle of
$G$ having at least four edges.  The seminal Strong Perfect Graph
Theorem of Chudnovsky, Robertson, Seymour, and
Thomas~\cite{ChudnovskyRST06,ChudnovskyS09}, conjectured by Berge in
1960~\cite{Berge60,Berge61,Berge85}, states that a graph is perfect
\marginppp{p2c4} if and only if it has no odd hole or complement of an
odd hole, implying \marginppp{p2c5} that a perfect graph $G$ can be
recognized by detecting an odd hole in $G$ or its
complement~$\bar{G}$.  Based on the theorem, the first known
polynomial-time algorithms for recognizing perfect graphs
take~$O(n^{18})$~\cite{CornuejolsLV03} and
$O(n^9)$~\cite{ChudnovskyCLSV05} time.  The $O(n^9)$-time version can
be implemented to run in $O(n^{8.373})$ time via Boolean matrix
multiplications~\cite[\S6.2]{LaiLT20}.

Detecting a class of induced subgraphs can be more difficult than
detecting its counterpart that need not be
induced~\cite{DalirrooyfardV22}. For instance, detecting a path
spanning three prespecified vertices is tractable (via,
e.\,g.,~\cite{KawarabayashiKR12,RobertsonS95b}), but the {\em
  three-in-a-path} problem that detects an induced path spanning three
prespecified vertices is NP-hard~(see,
e.\,g.,~\cite{HaasH06,LaiLT20}).  Cycle detection has a similar
situation.  Detecting a cycle of length three, which has to be
induced, is the classical triangle detection problem that can also be
solved efficiently by Boolean matrix multiplications.  Although it is
tractable to detect a cycle of length at least four spanning two
prespecified vertices (also via,
e.\,g.,~\cite{KawarabayashiKR12,RobertsonS95b}), the {\em
  two-in-a-cycle} problem that detects a hole spanning two
prespecified vertices is NP-hard (and so are the corresponding
one-in-an-even-cycle and one-in-an-odd-cycle
problems)~\cite{Bienstock91,Bienstock92}.  See,
e.\,g.,~\cite[\S3.1]{RadovanovicTV21} for graph classes on which the
two-in-a-cycle problem is tractable.  Detecting a tree spanning an
arbitrary set of prespecified vertices is easy via computing the
connected components of $G$. Detecting an induced tree spanning an
arbitrary set of prespecified vertices is NP-hard~\cite{GolovachPL12}.
The \emph{three-in-a-tree} problem that detects an induced tree
spanning three prespecified vertices was first shown to be solvable in
$O(n^4)$ time~\cite{ChudnovskyS10} and then in $O(n^2\log^2 n)$
time~\cite{LaiLT20}.  The tractability of the
corresponding~$k$-in-a-tree problem for any fixed
\marginppp{p3c1} $k\geq 4$ is unknown. See~\cite{LiuT10} for an
$O(n^4)$-time algorithm for the $k$-in-a-tree problem in a graph of
girth at least $k$.

As for subgraph detection without the requirement of spanning
prespecified vertices, detecting a cycle is straightforward. Even and
odd cycles are also long known to be efficiently detectable
(see,~e.\,g.,~\cite{AlonYZ95,DahlgaardKS17,YusterZ97}). While a hole
can be detected (i.\,e., recognizing chordal graphs) in $O(n^2)$
time~\cite{RoseTL76,TarjanY84,TarjanY85} and found in $O(n^4)$
time~\cite{NikolopoulosP07}, detecting an odd (respectively, even)
hole is more difficult. There are early $O(n^3)$-time algorithms for
detecting odd and even holes in planar graphs~\cite{Hsu87,Porto92},
but the tractability of detecting an odd hole was open for decades
(see,~e.\,g.,~\cite{ChudnovskyS18,ConfortiCKV99,ConfortiCLVZ06}) until
the recent major breakthrough of Chudnovsky, Scott, Seymour, and
Spirkl~\cite{ChudnovskySSS20-jacm-odd-hole} showing that an odd hole
can be detected in $O(n^9)$ time.  Via the Strong Perfect Graph
Theorem
\marginppp{p3c2} this yields another $O(n^9)$-time algorithm to
recognize perfect graphs.  Their $O(n^9)$-time algorithm is later
implemented to run in $O(n^8)$ time~\cite{LaiLT20} and recently
improved to run in $O(n^7)$ time~\cite{ChiuLL22}, implying that the
state-of-the-art algorithm
\marginppp{p3c3} for recognizing perfect graphs runs in $O(n^7)$ time.
It is also known that a shortest odd hole can be found in $O(n^{14})$
time~\cite{ChudnovskySS21-shortest-odd-hole} and~$O(n^{13})$
time~\cite{ChiuLL22}.  As for detecting even holes, the first
polynomial-time algorithm, running in about $O(n^{40})$ time, appeared
in 1997~\cite{ConfortiCKV97,ConfortiCKV02a,ConfortiCKV02b}. It takes a
line of intensive efforts to bring down the complexity to
$O(n^{31})$~\cite{ChudnovskyKS05},
$O(n^{19})$~\cite{daSilvaV13},~$O(n^{11})$~\cite{ChangL15}, and
finally $O(n^9)$~\cite{LaiLT20}.  The tractability of finding a
shortest even hole, open for $16$
years~\cite{ChudnovskyKS05,Johnson05}, is resolved by~\cite{CheongL21}
showing with more
\marginppp{p3c5} careful analysis
\marginppp{p3c4} that the $O(n^{31})$-time algorithm
of~\cite{ChudnovskyKS05} actually outputs a shortest even hole if
there is one.  The complexity is recently reduced to
$O(n^{19})$~\cite{ChiuLL22}.  See~\cite{ChudnovskySS21}
(respectively,~\cite{CookS20}) for detecting an odd (respectively,
even) hole with a prespecified length lower bound.
See~\cite{AbrishamiCPRS21,ChudnovskyPPT20} for the first
polynomial-time algorithm for finding an independent set of maximum
weight in a graph having no hole of length at least five.
See~\cite{DalirrooyfardVW19} for upper and lower bounds on the
complexity of detecting an~$O(1)$-vertex induced subgraph.

The \emph{two-in-a-path} problem that detects an induced path spanning
two prespecified vertices is equivalent to determining whether the two
vertices are connected.  On the other hand, the corresponding
two-in-an-odd-path and two-in-an-even-path problems are
NP-hard~\cite{Bienstock91,Bienstock92}, although each of them admits
an~$O(n^7)$-time algorithm when $G$ is planar~\cite{KaminskiN12}.
See~\cite{EverettFSMPR97,FonluptU82,Meyniel87} for how an induced even
$uv$-path of $G$ affects whether 
\marginppp{p3c6} $G$ is perfect.  See~\cite{Kriesell01a} for a
conjecture by Erd\H{o}s on how an induced $uv$-path of $G$ affects the
connectivity between $u$ and~$v$ in $G$.  Finding a longest $uv$-path
in $G$ that has to (respectively, need not) be induced is
NP-hard~\cite[GT23]{GareyJ79}~(respectively,~\cite[ND29]{GareyJ79}).
See~\cite{GiacomoLM16,JaffkeKT20} for longest or long induced paths in
special graphs.
The presence of long induced paths in $G$ affects the tractability of
coloring $G$~\cite{GaspersHP18}.  See also~\cite{AbrishamiCPRS21} for
the first polynomial-time algorithm for finding a minimum
\marginppp{p3c7} vertex set intersecting all cycles of a graph having
no induced path of length at least five.  Detecting a
non-shortest~$uv$-path in~$G$ is easy.  A $k$-th shortest $uv$-path in
$G$ can also be found in near linear time~\cite{Eppstein98}.
See~\cite{HoangKSS13} for algorithms
\marginppp{p3c8} that list induced paths and induced cycles.
See~\cite[\S4]{ChenF07} for the parameterized complexity of detecting
an induced path with a prespecified length.  Detecting an induced
$uv$-path in a directed graph~$G$ is NP-complete (even if $G$ is
planar)~\cite{FellowsKMP95} and~$W[1]$-complete~\cite{HaasH06}.
However, the tractability of detecting a non-shortest induced
$uv$-path in an undirected graph~$G$ was unknown until the recent
result of Berger, Seymour, and Spirkl~\cite{BergerSS21-dm}.

Let $\|G\|$ denote the number of edges in $G$.  A path with
end-vertices $u$ and $v$ is a~{\em$uv$-path}.  If $P$ is a path with
$\{u,v\}\subseteq V(P)$, then let $P[u,v]$ denote the $uv$-path of
$P$. A $uv$-path~$P$ of~$G$ is \emph{shortest} if $G$ admits no
$uv$-path $Q$ with $\|Q\|<\|P\|$, so each shortest $uv$-path of~$G$ is
induced.  We call an induced~$uv$-path of~$G$ that is not a shortest
$uv$-path of~$G$ a~\emph{$uv$-trail} of $G$.  See
Figure~\ref{figure:figure1} for an example.  A graph admitting
no~$uv$-trail is \emph{$uv$-trailless}.  Berger, Seymour, and
Spirkl~\cite{BergerSS21-dm} gave the formerly only known
polynomial-time algorithm, running in~$O(n^{18})$ time, to either
output a $uv$-trail of $G$ or ensure that $G$ is~$uv$-trailless.
Their result leads to an $O(n^{21})$-time
algorithm~\cite{CookHPRSSTV21} to determine whether all holes of $G$
have the same length.  The $\tilde{O}$ notation hides an $O(\log^2 n)$
factor and~$\omega<2.373$ denotes the exponent of square-matrix
multiplication~\cite{AlmanW21,LeGall14,VassilevskaWilliams12}
throughout the paper.  We improve the time of finding a $uv$-trail to
$O(n^{4.75})$ as summarized in the following theorem, which
immediately reduces the $O(n^{21})$ time~\cite{CookHPRSSTV21} of
recognizing a graph with all holes the same length to $O(n^{7.75})$.

\begin{figure}
\centering
\includegraphics[width=.7\textwidth]{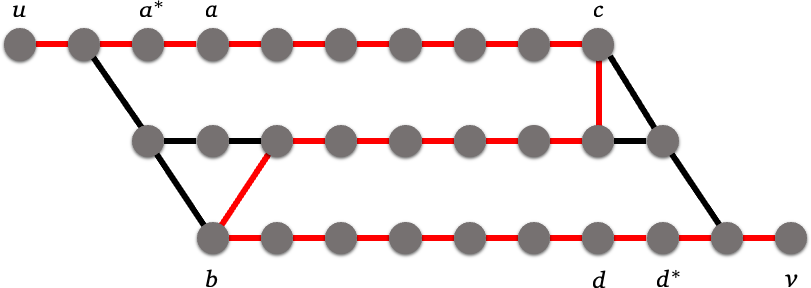}
\caption{The red $uv$-path $P$ is the only $uv$-trail of the
  $uv$-straight graph $G$. The twist pair of $P$ is~$(c,b)$. The twist
  of $P$ is $6$.  $P[a^*,c]$ and $P[b,d^*]$ form a pair of wings for
  the quadruple $(a,b,c,d)$ of $V(G)$ in~$G$.}
\label{figure:figure1}
\end{figure}

\begin{theorem}
\label{theorem:theorem1}
For any two vertices $u$ and $v$ of an $n$-vertex graph $G$, it takes
$\tilde{O}(n^{2\omega})$ time to either obtain a~$uv$-trail of $G$ or
ensure that $G$ is $uv$-trailless.
\end{theorem}

\paragraph*{Technical overview}
Berger et al.'s and our algorithms are based on the following
``guess-and-verify'' approach.  A subroutine $B$ taking an
$\ell$-tuple of $V(G)$ as the only argument is a
\emph{$uv$-trailblazer of degree~$\ell$} for~$G$ if running $B$ on all
$\ell$-tuples of $V(G)$ always reports a $uv$-trail of $G$ unless~$G$
is~$uv$-trailless.  We call an $\ell$-tuple of $V(G)$ on which $B$
reports a $uv$-trail of~$G$ a \emph{trail marker} for~$B$.
An~$O(f(n))$-time $uv$-trailblazer of degree $\ell$ for~$G$
immediately implies the following~$O(n^\ell\cdot f(n))$-time {\em
  trailblazing algorithm} for~$G$:
\marginppp{p4c1}
\begin{tabbing}
For each $\ell$-tuple $(a_1,\ldots,a_\ell)$ of $V(G)$,\\
\qquad\=if $B(a_1,\ldots,a_\ell)$ returns a $uv$-trail $P$ of $G$, then output $P$ and halt the algorithm;\\
\>otherwise, knowing that $(a_1,\ldots,a_\ell)$ is not a trail marker for $B$, proceed to the next iteration.\\
Report that $G$ is $uv$-trailless.
\end{tabbing}

A graph $H$ is \emph{$uv$-straight}~\cite{BergerSS21-dm} if
$\{u,v\}\subseteq V(H)$ and each vertex of $H$ belongs to at least one
shortest~$uv$-path of $H$.  For instance, the graph in
Figure~\ref{figure:figure1} is $uv$-straight.  Berger et al.'s
algorithm starts with an~$O(n^3)$-time preprocessing step (see
Lemma~\ref{lemma:lemma1}) that either reports a~$uv$-trail of $G$ or
obtains a $uv$-straight graph $H$ with $V(H)\subseteq V(G)$ such that
\begin{itemize}[leftmargin=*]
\item a~$uv$-trail of~$G$ can be obtained from a $uv$-trail of $H$
  in~$O(n^2)$ time and
\item if $H$ is $uv$-trailless, then so is $G$.
\end{itemize}
If no $uv$-trail is reported by the preprocessing, then the main
procedure runs an~$O(n^{18})$-time trailblazing algorithm on the
$uv$-straight graph $H$ based on an $O(n^4)$-time degree-14
$uv$-trailblazer for~$H$. As for postprocessing, if a $uv$-trail of
$H$ is obtained by the main procedure, then report a~$uv$-trail of $G$
obtainable in $O(n^2)$ time as ensured by the preprocessing.
Otherwise, report that~$G$ is $uv$-trailless.

Our $O(n^{4.75})$-time algorithm adopts the preprocessing and
postprocessing steps of Berger et al., while reducing the
preprocessing time from $O(n^3)$ to $O(n^\omega)$ (see
Lemma~\ref{lemma:lemma5}).  For the benefit of the main procedure, we
run a second preprocessing step, taking $O(n^{4.75})$ time via the
witness-matrix technique
\marginppp{p5c1} of Galil and Margalit~\cite{GalilM93}, to compute a
static data structure from which a pair of ``wings'' that are some
disjoint paths in $H$, if any, for each quadruple of~$V(H)$ can be
obtained in $O(n)$ time (see Lemma~\ref{lemma:lemma6}).  Our main
procedure is also a trailblazing algorithm, based on a
faster~$uv$-trailblazer of a lower degree for $H$: We reduce the time
from $O(n^4)$ to $O(n^2\log^2 n)$ and largely bring down the degree
from $14$ to $2$. Thus, the main procedure runs in $O(n^4\cdot\log^2
n)$ time, even faster than the second preprocessing step.

The key to our improved $uv$-trailblazer is a new observation,
described by Lemma~\ref{lemma:lemma4}, on any shortest~$uv$-trail~$P$
of a $uv$-straight graph $G$.  Specifically, Berger et al.'s
degree-$14$ $uv$-trailblazer seeks in $O(n^4)$ time a $uv$-trail $P$
of $G$ that consists of
\begin{itemize}[leftmargin=*]
\item a shortest $us$-path $S$ of $G$ containing $7$ prespecified
  vertices and a shortest $tv$-path $T$ of $G$ containing another $7$
  prespecified vertices such that $S$ and~$T$ are disjoint and
  nonadjacent
\marginppp{p5c2} in $G$ and
\item a shortest $st$-path~$Q$ of~$G_{S,T}=G-(N_G[S\cup T]\setminus
  N_G[\{s,t\}])$.
\end{itemize}
The $14$ vertices are to ensure that $s$ and $t$ are connected in
$G_{S,T}$.  Lemma~\ref{lemma:lemma4} implies that much fewer
prespecified vertices on~$S$ and $T$ suffice to guarantee that $s$ and
$t$
\marginppp{p5c3} are connected in $G_{S,T}$.  To illustrate the
usefulness of Lemma~\ref{lemma:lemma4}, we show
in~\S\ref{section:section2} that three lemmas of Berger et
al.~\cite{BergerSS21-dm} (i.\,e.,
Lemmas~\ref{lemma:lemma1},~\ref{lemma:lemma2}, and~\ref{lemma:lemma3})
together with Lemma~\ref{lemma:lemma4} already yield an
$O(n^2)$-time~$uv$-trailblazer of degree $4$ for $G$, leading to a
simple $O(n^6)$-time trailblazing algorithm on $G$.  More precisely,
as in the example of
\marginppp{p5c4} Figure~\ref{figure:figure1}, let
$\{a,b,c,d\}\subseteq V(P)$ for a shortest $uv$-trail $P$ of $G$ with
$d_P(u,a)\leq d_P(u,c) <d_P(u,b)\leq d_P(u,d)$, $d_G(u,a)=d_G(u,b)$,
and $d_G(u,c)=d_G(u,d)$ such that $n\cdot
(d_G(u,c)-d_G(u,a))+d_P(a,b)+d_P(c,d)$ is maximized.  Due to the
symmetry between $u$ and $v$ in $G$, Lemma~\ref{lemma:lemma4}
guarantees an~$O(n^2)$-time obtainable $uv$-trail of $G$ that contains
the precomputed pair of ``wings'' for the $4$-tuple $(a,b,c,d)$ (see
Lemma~\ref{lemma:lemma2}), implying that $(a,b,c,d)$ is a trail marker
for an $O(n^2)$-time $uv$-trailblazer for~$G$.

Our proof of Theorem~\ref{theorem:theorem1}
in~\S\ref{section:section3} further displays the usefulness of
Lemma~\ref{lemma:lemma4}.  We show that the aforementioned vertices
$a$ and $b$ in $P$ actually form a trail marker $(a,b)$ for
an~$O(n^2\log^2 n)$-time~$uv$-trailblazer for~$G$.  Dropping both $c$
and $d$ from the trail marker~$(a,b,c,d)$ of \S\ref{section:section2}
inevitably increases the time of the~$uv$-trailblazer for $G$.  We
manage to keep the time of a degree-two $uv$-trailblazer as low
as~$O(n^2\log^2 n)$ via the dynamic data structure of Holm,
de~Lichtenberg, and Thorup~\cite{HolmdT01} supporting efficient edge
updates and connectivity queries for $G$ (see
Lemma~\ref{lemma:lemma7}).  To make our proof of
Theorem~\ref{theorem:theorem1} in~\S\ref{section:section3}
self-contained, a simplified proof of Lemma~\ref{lemma:lemma3} is
included in~\S\ref{section:section2}.  Since Lemmas~\ref{lemma:lemma1}
and~\ref{lemma:lemma2} are implied by Lemmas~\ref{lemma:lemma5}
and~\ref{lemma:lemma6}, which are proved in~\S\ref{section:section3},
our proof for the $O(n^6)$-time algorithm in~\S\ref{section:section2}
is also self-contained.

\section[]{A simple $\boldsymbol{O(n^6)}$-time algorithm}
\label{section:section2}

Let $G$ be a connected graph containing vertices $u$ and $v$.  For any
vertices $x$ and $y$ of~$G$, let~$d_G(x,y)=\|P\|$ for a
shortest~$xy$-path~$P$ of~$G$.  Let $h(x)=d_G(u,x)$ be the {\em
  height} of a vertex~$x$ in $G$.  If $xy$ is an edge of~$G$, then
$|h(x)-h(y)|\leq 1$.  For an $H\subseteq G$,~(i)
let~$G-H=G[V(G)\setminus V(H)]$,~(ii) let $N_G(H)$ consist of the
vertices $y\in V(G-H)$ adjacent to at least one vertex of~$H$ in $G$,
and (iii) let $N_G[H]=N_G(H)\cup V(H)$.  For an~$x\in V(G)$,
let~$G-x=G-\{x\}$, let~$N_G(x)=N_G(\{x\})$, and
let~$N_G[x]=N_G[\{x\}]$.  $X$ and $Y$ are \emph{adjacent}
(respectively, \emph{anticomplete}) in~$G$ if~$N_G(X)\cap
V(Y)\ne\varnothing$ (respectively,~$N_G[X]\cap V(Y)=\varnothing$).

\begin{lemma}[{Berger et al.~\cite[Lemma~2.2]{BergerSS21-dm}}]
\label{lemma:lemma1}
For any vertices $u$ and $v$ of an $n$-vertex connected graph $G_0$,
it takes~$O(n^3)$ time to obtain (1) a~$uv$-trail of $G_0$ or (2) a
$uv$-straight graph $G$ with~$V(G)\subseteq V(G_0)$ such that~
\begin{enumerate}[label={(\alph*)}]
\item a~$uv$-trail of $G_0$ is $O(n^2)$-time obtainable from that of
  $G$ and
\item if $G$ is $uv$-trailless, then so is~$G_0$.
\end{enumerate}
\end{lemma}

A path of $G$ is \emph{monotone}~\cite{BergerSS21-dm} if all of its
vertices have distinct heights in $G$.  A monotone~$xy$-path of~$G$ is
a shortest $xy$-path of $G$. The converse may not hold.  A shortest
$xy$-path of a $uv$-straight
\marginppp{p6c1} graph $G$ with $\{x,y\}\cap \{u,v\}\ne\varnothing$ is
monotone.  A monotone $ca^*$-path
\marginppp{p6c2} $W_1$ of $G$ containing a vertex $a$ and a
monotone~$bd^*$-path $W_2$ of $G$ containing a vertex $d$ with
\[
h(a^*)+1=h(a)=h(b)\leq h(c)=h(d)=h(d^*)-1
\] 
form a pair $(W_1,W_2)$ of \emph{wings} for the quadruple $(a,b,c,d)$
of $V(G)$ in $G$ if
\[
d_{G[W_1\cup W_2]}(a^*,d^*)> \|W_1\|+\|W_2\|,
\] 
that is, $W_1-c$ (respectively, $W_1$) and $W_2$ (respectively,
$W_2-b$) are anticomplete in $G$.  A
\marginppp{p6c3} quadruple $(a,b,c,d)$ of $V(G)$ is \emph{winged} in
$G$ if $G$ admits a pair of wings for $(a,b,c,d)$.  See
Figure~\ref{figure:figure1} for an example.

\begin{lemma}[{Implicit in Berger et al.~\cite[Lemma~2.1]{BergerSS21-dm}}]
\label{lemma:lemma2}
It takes $O(n^6)$ time 
\marginppp{p6c4} to compute a data structure from which the following
statements hold for any quadruple $(a,b,c,d)$ of $V(G)$ for an
$n$-vertex graph $G$:
\begin{enumerate}
\item It takes $O(1)$ time to determine whether $(a,b,c,d)$ is winged
  in $G$.
  
\item If $(a,b,c,d)$ is winged in $G$, then it takes $O(n)$ time to
  obtain a pair of wings for $(a,b,c,d)$ in $G$.
\end{enumerate}
\end{lemma}
We comment that the proof of~\cite[Lemma~2.1]{BergerSS21-dm} is easily
adjustable into one for Lemma~\ref{lemma:lemma2}.  Moreover,
see~\S\ref{section:section3} for the proof of
Lemma~\ref{lemma:lemma6}, which implies and improves upon
Lemma~\ref{lemma:lemma2}.

Let $P$ be a $uv$-path of a $uv$-straight graph $G$.
The \marginppp{p6c6} \emph{twist pair} of $P$ is the vertex pair
$(s,t)$ of $P$ such that $P[u,s]$ and $P[t,v]$ are the
\marginppp{p6c5} maximal monotone prefix and suffix of $P$.  The
\emph{twist}~\cite{BergerSS21-dm} of~$P$ is~$h(s)-h(t)$ for the twist
pair $(s,t)$ of $P$.  See also Figure~\ref{figure:figure1} for an
example. If $(s,t)$ is the twist pair of a $uv$-path~$P$ of $G$,
then~$P[u,s]$ and~$P[t,v]$ are disjoint if and only if $P$ is a
non-shortest $uv$-path of $G$.  The next lemma is also needed
in~\S\ref{section:section3}. To make our proof of
Theorem~\ref{theorem:theorem1} in~\S\ref{section:section3}
self-contained, we include a proof of Lemma~\ref{lemma:lemma3}
simplified from that of Berger et al.~\cite[Lemma~2.3]{BergerSS21-dm}.

\begin{lemma}[{Berger et al.~\cite[Lemma~2.3]{BergerSS21-dm}}]
\label{lemma:lemma3}
If $(s,t)$ is the twist pair of a shortest $uv$-trail $P$ of a
$uv$-straight graph $G$, then $h(s)\geq h(x)\geq h(t)$ holds for each
vertex $x$ of $P[s,t]$.
\end{lemma}

\begin{proof}
Let~$I=V(P[s,t])\setminus \{s,t\}$.
\marginppp{p6c7} Let $s^*$ (respectively, $t^*$) be the neighbor of
$s$ (respectively,~$t$) in $P[s,t]$.  By definition of $(s,t)$, we
have $h(s^*)\leq h(s)$ and $h(t^*)\geq h(t)$.  If~$I=\varnothing$,
then~$(s^*,t^*)=(t,s)$ implies the lemma.  Otherwise, it suffices to
prove $h(s)\geq h(x)\geq h(t)$ for each~$x\in I$.  If~$h(x)>h(s)$ were
true for the $x\in I$ maximizing $n\cdot h(x)+d_{P[s,t]}(x,t)$, then
the concatenation of $P[u,x]$ and a shortest $xv$-path of~$G$ is
a~$uv$-trail of~$G$ shorter than $P$.
\marginppp{p6c8} If $h(x)<h(t)$ were true for the~$x\in I$ minimizing
$n\cdot h(x)+d_{P[s,t]}(x,t)$, then the concatenation of a shortest
$ux$-path of $G$ and $P[x,v]$ is a $uv$-trail of~$G$ shorter than~$P$.
\marginppp{p6c9}
\end{proof}

\begin{figure}
\centering
\includegraphics[width=.7\textwidth]{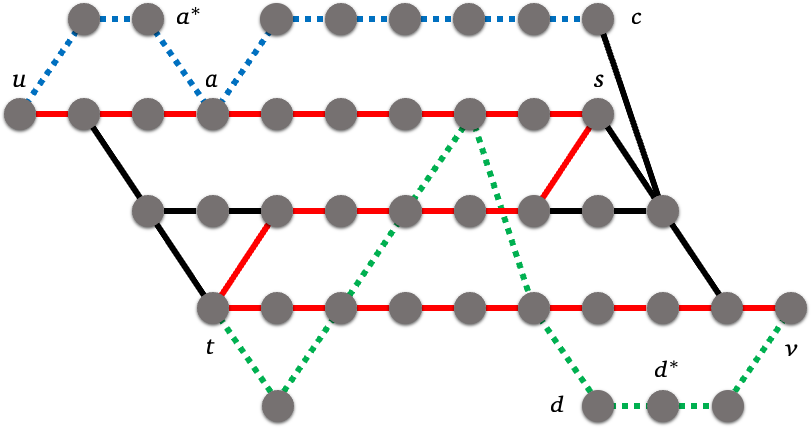}
\caption{The blue dotted $uc$-path is a sidetrack $S$ for the red $uv$-trail
  $P$ of the $uv$-straight graph $G$.  Each of $P[t,v]$ and the green dotted
  $tv$-path can be a monotone~$tv$-path $T$ satisfying
  Condition~S\ref{S1}.
}
\label{figure:figure2}
\end{figure}

Let $P$ be a $uv$-trail of a $uv$-straight graph $G$ with twist pair
$(s,t)$. Lemma~\ref{lemma:lemma3} implies $h(t)\leq h(s)$.  A monotone
$uc$-path $S$ of $G$ with $h(c)=h(s)$ is a \emph{sidetrack} for $P$ if
the
\marginppp{p6c10} following \emph{Conditions~S} hold.
\begin{enumerate}[label={\emph{S}\arabic*:}, ref={\arabic*}]
\item 
\label{S1} 
$G$ contains a monotone $tv$-path $T$ with $d_{G[S\cup T]}(u,v)>
\|S\|+\|T\|$.

\item 
\label{S2}
$S$ contains the vertex $a$ of $P[u,s]$ with $h(a)=h(t)$.
\end{enumerate}
We comment that existence
\marginppp{p6c11} and unique of $a$ in Condition~S\ref{S2} follows
from $h(t)\leq h(s)$ and the fact that $P[u,s]$ is monotone
\marginppp{p6c12} by definition of twist pair $(s,t)$.
Condition~S\ref{S1} is equivalent to the statement that $S-c$
(respectively, $S$) and $T$ (respectively, $T-t$) are anticomplete in
$G$.  Let $a^*$ be the vertex of the monotone $uc$-path $S$ with
$h(a^*)=h(a)-1$.  Let $dd^*$ be the edge of the
\marginppp{p7c2} monotone $tv$-path $T$ with~$h(s)=h(d)=h(d^*)-1$.
Condition~S\ref{S1} implies that $S[a^*,c]$ and $T[t,d^*]$ form a pair
of wings for~$(a,t,c,d)$ in~$G$.  See Figure~\ref{figure:figure2}
\marginppp{p7c1} for an example.  The key to our largely
improved~$uv$-trailblazers in~\S\ref{section:section2}
and~\S\ref{section:section3} is the following lemma, whose proof is
illustrated in Figure~\ref{figure:figure3}.

\begin{figure}
\centering
\includegraphics[width=.7\textwidth]{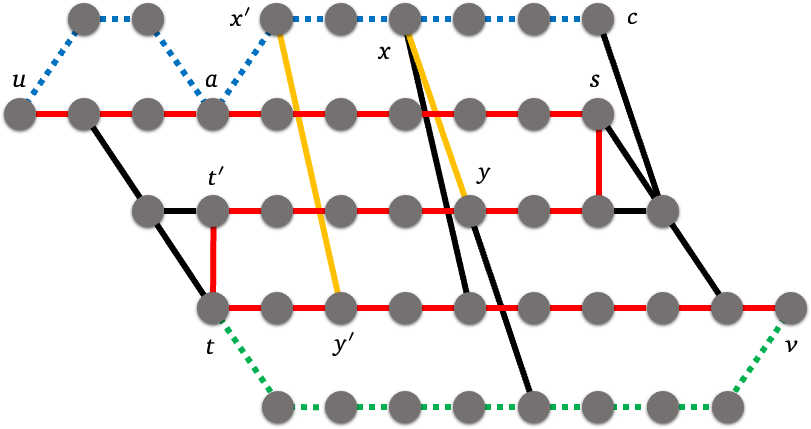}  
\caption{An illustration for the proof of Lemma~\ref{lemma:lemma4}.
  The red path denotes a shortest $uv$-trail $P$ of the~$uv$-straight
  graph $G$. The blue dotted monotone path denotes a sidetrack $S$ for
  $P$. The green dotted path denotes a monotone path $T$ satisfying
  Condition~S\ref{S1}.}
\label{figure:figure3}
\end{figure}

\begin{lemma}
\label{lemma:lemma4}
If $S$ is a sidetrack for a shortest $uv$-trail $P$ of a $uv$-straight
graph $G$ with twist pair $(s,t)$, then
\[
d_{G[S\cup P[s,t]]}(u,t)\geq d_P(u,t).
\]
\end{lemma}

\begin{proof}
Let $T$ be a monotone $tv$-path of $G$ with $d_{G[S\cup T]}(u,v)>
\|S\|+\|T\|$ by
\marginppp{p7c3} Condition~S\ref{S1}.  Assume for contradiction
\marginppp{p7c4} that there exists a shortest $ut$-path $Q$ of
$G[S\cup P[s,t]]$ with
\begin{equation}
\label{equation:eq1}
\|Q\|<d_P(u,t),
\end{equation}
implying $\|R\|<\|P\|$ for a shortest $uv$-path $R$ of $G[Q\cup T]$.
Since $R$ is an induced $uv$-path of $G$ shorter than $P$, $R$ is
monotone.  Since Condition~S\ref{S1} implies $t\notin V(S)$,
\marginppp{p7c5,9}
there exists an edge $xy$ of $Q$ with $x\in V(S)$ and~$y\in
V(P[s,t])$ that minimizes~$d_{P[s,t]}(y,t)$.  
\marginppp{p7c6,8} We have $\{x,y\}\subseteq V(R)$ or else $R$
deviates from the induced path $Q$ at a vertex $q$ of $Q[u,x]=S[u,x]$
and enters $T$
\marginppp{p7c7} (by $V(R)\subseteq V(Q\cup T)$) at a vertex $b$ with
height $h(q)+1$ (by monotonicity of $R$), violating
Condition~S\ref{S1}.
\marginppp{p7c10} Since $d_R(u,x)<d_R(u,y)$, we have
\begin{equation}
\label{equation:eq2}
h(x)+1=h(y).
\end{equation}
By Equation~\eqref{equation:eq1}, 
\marginppp{p7c11--13} a shortest $uv$-path $R'$ of $G[Q\cup T']$ with
$T'=P[t,v]$ is monotone. Hence,
\begin{equation}
\label{equation:eq3}
h(x')+1=h(y')
\end{equation}
holds for an edge $x'y'$ of $R'$ with $x'\in V(Q-t)$ and $y'\in
V(T')$,
\marginppp{p7c14--18} implying $x'\notin V(P[s,t])$ by
Lemma~\ref{lemma:lemma3}.  Thus, $x'\in V(S[u,x])$.
Condition~S\ref{S1} implies $y'\ne t$.  Condition~S\ref{S2} and
Equation~\eqref{equation:eq3} imply $h(y')\ne h(t)+1$.  Therefore, we
have $h(y')\geq h(t)+2$, implying $h(x)\geq h(t)+1$ (by
Equation~\eqref{equation:eq3}) and
\marginppp{p7c19}
\begin{equation}
\label{equation:eq4}
d_G(t',y')\geq 2
\end{equation}
for the vertex~$t'\in V(P[y,t])$ with~$h(t')=h(t)$ that minimizes $d_P(t',y)$.
\marginppp{p7c20,21} Since $h(x)\geq h(t)+1$ and by choices of $y$ and
$t'$, the concatenation $Q'$ of
\begin{enumerate}[leftmargin=*, label={(\roman*)}]
\item  
a shortest~$ut'$-path of~$G$ and
\item 
$P[t',y]$
\end{enumerate}
is an induced $uy$-path of $G$ with $\|Q'\|\leq
\|P[u,a]\|+\|P[s,t]\|$.  By
Equations~\eqref{equation:eq2},~\eqref{equation:eq3},
and~\eqref{equation:eq4}, a shortest $xv$-path $Q''$ of $G[S[x,x']\cup
  P[y',v]]$ satisfies $\|Q''\|\leq \|P[a,s]\|+\|P[t,v]\|-2$. Thus,
$P'=Q'\cup yx\cup Q''$
\marginppp{p7c22} is a $uv$-path of $G$ with $\|P'\|<\|P\|$.
\marginppp{p8c1} By definitions of the $ut$-path $Q$ and the edge $xy$
of $Q$, $Q'$ is anticomplete to $S[x',x]-x$ in $G$.  By
Equation~\eqref{equation:eq4}, $Q'$ is anticomplete to $P[y',v]$ in
$G$. Hence, $P'$ is an induced $uv$-path of $G$.  By
Equation~\eqref{equation:eq2} and $d_{P'}(u,x)>d_{P'}(u,y)$, $P'$ is
a~$uv$-trail of $G$ shorter than $P$, contradiction.
\end{proof}

We are ready to describe and justify an $O(n^6)$-time algorithm that
either reports a~$uv$-trail of $G$ or ensures that~$G$
is~$uv$-trailless. Let $G$ be connected without loss of generality.

\paragraph*{Our $\boldsymbol{O(n^6)}$-time algorithm}
Let $G_0$ be the input $n$-vertex graph.  Apply
Lemma~\ref{lemma:lemma1} in $O(n^3)$ time to either report a
$uv$-trail of $G_0$ as stated in Lemma~\ref{lemma:lemma1}(1) or obtain
a~$uv$-straight graph $G$
\marginppp{p8c2} satisfying Conditions~(a) and~(b) of
Lemma~\ref{lemma:lemma1}(2). If no~$uv$-trail is reported in the
previous step, then apply Lemma~\ref{lemma:lemma2} to obtain the data
structure $D$ for the winged quadruples of $G$ in~$O(n^6)$ time.  By
Conditions~(a) and~(b) of
\marginppp{p8c3} Lemma~\ref{lemma:lemma1}(2), it remains to show
an~$O(n^2)$-time degree-$4$~$uv$-trailblazer for the $uv$-straight
graph $G$, which immediately leads to an $O(n^6)$-time trailblazing
algorithm that either reports a $uv$-trail of $G$ or ensures that $G$
is $uv$-trailless.

Let $B$ be the following $O(n^2)$-time subroutine, taking a
quadruple~$(a,b,c,d)$ of $V(G)$ as the argument: Determine in $O(1)$
time from the data structure $D$ whether $(a,b,c,d)$ is winged in $G$.
If not, then exit. Otherwise, obtain in $O(n)$ time from $D$ a
pair~$(W_1,W_2)$ of wings for~$(a,b,c,d)$ in $G$.  Since $G$
is~$uv$-straight,
\marginppp{p8c4} each monotone $xy$-path $Q$ of $G$ with $h(x)\leq
h(y)$ is contained by the monotone $uv$-path $P\cup Q\cup R$, where
$P$ is an arbitrary monotone $ux$-path of $G$ and $R$ is an arbitrary
monotone $yv$-path of $G$.  Thus, it takes~$O(n^2)$ time to obtain a
monotone~$uc$-path~$S$ of $G$ containing $W_1$ and a monotone
$bv$-path~$T$ of~$G$ containing~$W_2$.  Obtain in~$O(n^2)$ time the
subgraph $G_{c,b}$ of $G$ induced by
\[
\{x\in V(G):h(b)\leq h(x)\leq h(c)\}\setminus ((N_G[S-c]\cup
N_G[T-b])\setminus\{c,b\}).
\]
If $c$ and $b$ are not connected in $G_{c,b}$, then exit.  Otherwise,
report the concatenation $P_{c,b}$ of~(i) the~$uc$-path~$S$, (ii) a
shortest $cb$-path of $G_{c,b}$, and (iii) the $bv$-path $T$.

By definition of $S$, $T$, and $G_{c,b}$, the $uv$-path $P_{c,b}$ of
$G$ reported by $B(a,b,c,d)$ is induced in~$G$, which is not monotone
by $h(b)\leq h(c)$.  Thus, $P_{c,b}$ is a $uv$-trail of~$G$.  

Let $P$ be an arbitrary unknown shortest $uv$-trail of~$G$ with twist
pair $(s,t)$.  By Lemma~\ref{lemma:lemma3}, we have $h(t)\leq h(s)$.
Let~$a$ (respectively,~$d$) be the vertex of the monotone $P[u,s]$
(respectively,~$P[t,v]$) with $h(a)=h(t)$ (respectively,~$h(d)=h(s)$),
whose existence and uniqueness
\marginppp{p9c1} are due to the fact that $P[u,s]$ and $P[t,v]$ are
monotone.  See Figure~\ref{figure:figure4} for an illustration.  The
rest of the section shows that $(a,t,s,d)$ is a trail marker for~$B$.

\begin{figure}
\centering
\includegraphics[width=.7\textwidth]{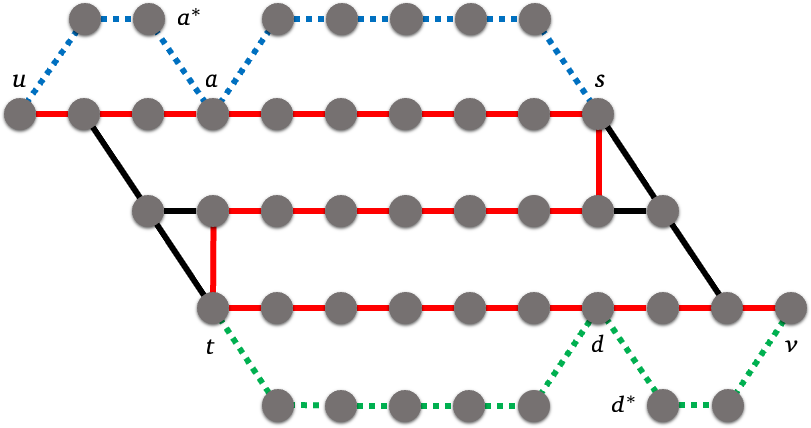}
\caption{An illustration for the proof that $B$ is a $uv$-trailblazer
  of degree four. The red path denotes a shortest $uv$-trail of the
  $uv$-straight graph $G$.  The blue and green dotted paths denote a
  monotone~$us$-path and a monotone $tv$-path of $G$ containing a
  precomputed pair of wings for~$(a,t,s,d)$ that need not coincide
  with $P$ except at $a$, $t$, $s$, and $d$.}
\label{figure:figure4}
\end{figure}

Let $a^*$ be the neighbor of $a$ in $P[u,a]$ 
\marginppp{p9c2} and $d^*$ be the neighbor of $d$ in $P[d,v]$.
$P[a^*,s]$ and~$P[t,d^*]$ form a pair of wings for~$(a,t,s,d)$ in
$G$. Thus, the quadruple $(a,t,s,d)$ is winged in $G$.  Let
$(W_1,W_2)$
\marginppp{p9c3} be a pair of wings for $(a,t,s,d)$.  The
monotone~$us$-path~$S$ of $G$ containing~$W_1$ is a sidetrack for $P$,
since the monotone $tv$-path $T$ of~$G$ containing $W_2$ satisfies
Conditions~S\ref{S1} and~S\ref{S2} for~$S$.  Due to the symmetry
\marginppp{p9c4} between~$u$ and $v$ in the undirected $uv$-straight
graph $G$ with respect to the height function that maps each vertex to
its distance to $v$ in $G$, the monotone $vt$-path~$T$ of the
$vu$-straight graph $G$ is also a sidetrack for the shortest
$vu$-trail $P$ of $G$ with twist pair~$(t,s)$, since the
monotone~$su$-path~$S$ of $G$ satisfies Conditions~S\ref{S1}
and~S\ref{S2} for $T$.  Lemma~\ref{lemma:lemma3} guarantees $h(t)\leq
h(x)\leq h(s)$ for each vertex $x$ of $P[s,t]$.  By
Lemma~\ref{lemma:lemma4},~$P[s,t]-\{s,t\}$ is anticomplete to both
$S-s$ and $T-t$, implying that~$P[s,t]$ is a path of $G_{s,t}$.
Since~$s$ and $t$ are connected in~$G_{s,t}$, the subroutine
call~$B(a,t,s,d)$ outputs a~$uv$-trail~$P_{s,t}$ of $G$ in~$O(n^2)$
time. Hence,~$(a,t,s,d)$ is indeed a trail marker of~$B$.

As a matter of fact, $P_{s,t}$ is a shortest $uv$-trail of $G$ due to the 
\marginppp{p9c5} fact that $\|P_{s,t}\|=\|P\|$. Since the
preprocessing and postprocessing may ruin the shortestness of the
reported $uv$-trail, we have an $O(n^6)$-time algorithm on
an~$n$-vertex general (respectively, $uv$-straight) graph $G$ that
either reports a general (respectively, shortest) $uv$-trail of $G$ or
ensures that $G$ is $uv$-trailless.

\section[]{An $\boldsymbol{O(n^{4.75})}$-time algorithm}
\label{section:section3}
This section gives a self-contained proof of
Theorem~\ref{theorem:theorem1}.  The \emph{product} of $m\times m$
Boolean matrices $A$ and~$B$ is the $m\times m$ Boolean matrix $C$
such that $C(i,k)=\mathit{true}$ if and only
if~$A(i,j)=B(j,k)=\mathit{true}$ holds for an index~$j$.  The
following lemma implies and improves upon Lemma~\ref{lemma:lemma1},
which states
\marginppp{p9c6} that it takes $O(n^3)$ time to obtain a $uv$-trail of
$G$ from a $uv$-trail of $H$.

\begin{lemma}
\label{lemma:lemma5}
For any vertices $u$ and $v$ of an $n$-vertex connected graph $G$, it
takes $O(n^\omega)$ time to obtain~(1) a~$uv$-trail of $G$ or (2) a
$uv$-straight graph $H$ with $V(H)\subseteq V(G)$ such that (a) a
$uv$-trail of $G$ can be obtained from a~$uv$-trail of $H$ in $O(n^2)$
time and (b) if $H$ is~$uv$-trailless, then so is $G$.
\end{lemma}

\begin{proof}
We adopt the proof of Berger et al.~\cite[Lemma~2.2]{BergerSS21-dm}
with slight simplification and improvement.  It takes $O(n^2)$ time
via, e.\,g.,
\marginppp{p10c1,2} breadth-first search to obtain a maximal set
$F\subseteq V(G)$ such that $G[F]$ is~$uv$-straight.  If $F=V(G)$,
then the lemma is proved by returning $H=G$.  The rest of the proof
assumes $F\subsetneq V(G)$.  It takes~$O(n^\omega)$ time to determine
whether some connected component~$K$ of $G-F$ admits an
$\{x,y\}\subseteq N_G(K)$
\marginppp{p10c6} with $xy\notin E(G)$ and $h(x)<h(y)$ via finding a
triangle intersecting exactly two vertices of $F$ in the following
graph $A$
\marginppp{p10c3,4} such that each vertex of $A$ corresponds to a
vertex $x$ in $F$ or a connected component $K$ of $G-F$:
\begin{itemize}
\item A vertex $x$ in $F$ is adjacent to a connected component $K$ of
  $G-F$ if and only if $x\in N_G(K)$.
\item Distinct connected components of $G-F$ are nonadjacent in $A$.
\item Distinct vertices $x$ and $y$ of $F$ are adjacent in $A$ if and
  only if $xy\notin E(G)$ and $h(x)\ne h(y)$.
\end{itemize}
If there is such a $(K,x,y)$, then a shortest $uv$-path of $G[P_x\cup
  K\cup P_y]$ for any monotone~$ux$-path~$P_x$ and $yv$-path $P_y$ of
$G$ is a $uv$-trail of~$G$ (by the maximality of
\marginppp{p10c5} $F$) obtainable in~$O(n^2)$ time, proving the lemma.
Otherwise, let $H$ be the union of the~$uv$-straight graph~$G[F]$ and
the $O(n^\omega)$-time obtainable graph $H'$ with $V(H')=F$ (via
contracting each connected component of $G-F$ into a single vertex and
then squaring the adjacency matrix) such that distinct vertices $x$
and~$y$ are adjacent in~$H'$ if and only if $\{x,y\}\subseteq N_G(K)$
holds for a connected component $K$ of~$G-F$. Observe that each
edge~$xy$ of~$H'$ with~$h(x)\ne h(y)$ is also an edge of $G[F]$.
By~$|h(x)-h(y)| \leq 1$ for all edges~$xy$ of $H'$, $H$
remains~$uv$-straight and~$d_H(u,x)=h(x)$ holds for each $x\in F$.  To
see Condition~(a), for any given $uv$-trail $Q$ of~$H$, let~$P$ be an
$O(n^2)$-time obtainable non-monotone $uv$-path of~$G$ obtained from
$Q$ by replacing each edge~$xy$ of~$Q$ not in $G[F]$ with a shortest
$xy$-path~$P_{xy}$ of $G-(F\setminus \{x,y\})$.  If $P$ were not
induced, then there is an edge~$zz'$ of $G[P]$ not in $P$ with~$z\in
V(P_{xy})$ and $z'\in V(P_{x'y'})$ for distinct edges $xy$ and $x'y'$
of~$Q$ that are not in $G[F]$.  Thus, $\{x,y,x',y'\}\subseteq N_G(K)$
holds for some connected component~$K$ of~$G-F$.  By definition
of~$H'$, $H[\{x,y,x',y'\}]$ is complete, contradicting that $Q$ is an
induced path of $H$. Thus, $P$ is a $uv$-trail of $G$, proving
Condition~(a).  As for Condition~(b), let~$P$ be a $uv$-trail of
$G$. For any distinct vertices $x$ and~$y$ of $P$ such that~$P[x,y]$
is a maximal subpath of $P$ contained by~$G[\{x,y\}\cup K]$ for some
connected component~$K$ of~$G-F$, $P[x,y]$ is an induced $xy$-path of
$G[\{x,y\}\cup K]$. The path $Q$ obtained from~$P$ by replacing each
such~$P[x,y]$ by the edge $xy$ of $H'$ is an induced $uv$-path of
$H$. If $Q$ were a shortest~$uv$-path of $H$, then~$|h(x)-h(y)|=1$
holds for each edge $xy$ of $Q$, implying that each edge $xy$ of $Q$
is an edge of $P$, contradicting that $P$ is a $uv$-trail of $G$.
\end{proof}

The bottleneck of our algorithm for Theorem~\ref{theorem:theorem1}
comes from the following lemma, which implies and improves upon
Lemma~\ref{lemma:lemma2} that takes $O(n^6)$ time.

\begin{lemma}
\label{lemma:lemma6}
It takes $\tilde{O}(n^{2\omega})$ time to compute a data structure
from which the following statements hold for any quadruple $(a,b,c,d)$
of $V(G)$ for an $n$-vertex graph $G$:
\begin{enumerate}
\item It takes $O(1)$ time to determine whether $(a,b,c,d)$ is winged
  in $G$.
\item If $(a,b,c,d)$ is winged in $G$, then it takes $O(n)$ time to
  obtain a pair of wings for $(a,b,c,d)$ in $G$.
\end{enumerate}
\end{lemma}

\begin{proof}
The lemma holds clearly for the quadruples $(a,b,c,d)$ of $V(G)$
with~$h(c)\leq h(a)+1$.  The rest of the proof handles those with
$h(a)+2\leq h(c)$.  A pair of wings for such an $(a,b,c,d)$ must be
anticomplete in $G$.  It takes $O(n^4)$ time to obtain the $n^2\times
n^2$ Boolean matrix $A$ such that~$A((a,b),(c,d))=\mathit{true}$ if
and only if~(i)~$h(a)=h(b)\leq h(c)=h(d)\leq h(a)+1$ and~(ii) $G$
admits a pair of anticomplete wings for~$(a,b,c,d)$.  The transitive
closure $C=A^n$ of $A$ can be obtained in~$O(n^{2\omega}\cdot\log n)$
time via obtaining~$A^{2^i}$ in the $i$-th iteration.  That is, for
each $(a,b,c,d)$, we have~$C((a,b),(c,d))=\mathit{true}$ if and only
if (i) $h(a)=h(b)\leq h(c)=h(d)$ and~(ii)~$G$ admits a pair of
anticomplete wings for~$(a,b,c,d)$ in $G$.  Statement~1 is proved.
Statement~2 is immediate from the~$\tilde{O}(n^{2\omega})$-time
obtainable~$n^2\times n^2$ witness matrix~$W$ for~$C$ by,~e.\,g.,
Galil and Margalit~\cite{GalilM93}: if~$C((a,b),(c,d))=\mathit{true}$
and $h(a)+2\leq h(c)$, then~$W((a,b),(c,d))$ is a vertex pair~$(x,y)$
with $h(a)<h(x)<h(c)$ and
$C((a,b),(x,y))=C((x,y),(c,d))=\mathit{true}$.
\end{proof}
The next dynamic data structure for a graph supports efficient
edge updates and connectivity queries.

\begin{lemma}[{Holm, de~Lichtenberg, and Thorup~\cite{HolmdT01}}]
\label{lemma:lemma7}
There is a data structure for an initially empty $n$-vertex graph that
supports each edge insertion and edge deletion in amortized~$O(\log^2
n)$ time and answers whether two vertices are connected in $O(\log
n/\log\log n)$ time.
\end{lemma}
We are ready to prove Theorem~\ref{theorem:theorem1}. Assume without
loss of generality that $G$ is connected.

\paragraph*{Our $\boldsymbol{O(n^{4.75})}$-time algorithm}

Apply Lemma~\ref{lemma:lemma5} in $O(n^\omega)$ time to either report
a $uv$-trail of~$G$ as in Lemma~\ref{lemma:lemma5}(1) or make~$G$
a~$uv$-straight graph satisfying Conditions~(a) and~(b) of
Lemma~\ref{lemma:lemma5}(2) with respect to the original $G$.  If no
$uv$-trail is reported in the previous step, then apply
Lemma~\ref{lemma:lemma6} in~$\tilde{O}(n^{2\omega})$ time to obtain
the data structure~$D$ for the winged quadruples of $V(G)$ in $G$.  It
remains to show an~$O(n^2\log^2 n)$-time degree-two $uv$-trailblazer
for the $uv$-straight graph $G$ based on the precomputed $D$ which
already spends $O(n^{4.75})$ time.  We proceed in two phases. Phase~1
shows that we already have an $O(n^3)$-time degree-two
$uv$-trailblazer for $G$.  Phase~2 then reduces the time to
$O(n^2\log^2 n)$ via Lemma~\ref{lemma:lemma7}.

\subparagraph*{Phase~1} 
Let $B_1$ be the $O(n^3)$-time subroutine, taking a pair~$(a,b)$ of
$V(G)$ as the only argument, that runs the following $O(n^2)$-time
procedure for each vertex $c$ of $G$: Determine from $D$ in $O(n)$
time whether~$G$ admits a winged quadruple $(a,b,c,d_c)$ of $V(G)$ for
some~$d_c$. If not, then exit.  Otherwise, obtain from~$D$ in~$O(n)$
time a pair~$(W_1,W_2)$ of wings for an arbitrary winged
$(a,b,c,d_c)$.  Since $G$ is~$uv$-straight, it takes~$O(n^2)$ time to
obtain a monotone~$uc$-path $S_c$ of~$G$ containing $W_1$ and a
monotone~$bv$-path~$T_c$ of~$G$ containing~$W_2$.  Obtain in~$O(n^2)$
time the subgraph $G_c$ of $G$ induced by
\[
(\{x\in V(G):h(b)\leq h(x)\leq h(c)\}\setminus (N_G[S_c-c]\setminus
\{c\}))\cup V(T_c).
\]
If the vertices $c$ and $b$ are not connected in $G_c$, then exit.
Otherwise, report the $O(n^2)$-time obtainable concatenation $P_c$ of
the $uc$-path $S_c$ of $G$ and a shortest $cv$-path of $G_c$.

By definition of $S_c$, $T_c$, and $G_c$, the $uv$-path $P_c$ of $G$
reported by $B_1(a,b)$ for any $c$ is induced in $G$.  Since the
height of each neighbor of $c$ in $G_c$ is at most $h(c)$, $P_c$ is
not monotone.  Thus, $P_c$ is a $uv$-trail of~$G$.  Let~$P$ be an
arbitrary unknown shortest $uv$-trail of~$G$ with twist pair~$(s,t)$.
Let~$a$ (respectively,~$e$) be the vertex of the monotone $P[u,s]$
(respectively,~$P[t,v]$) with $h(a)=h(t)$ (respectively, $h(e)=h(s)$).
See Figure~\ref{figure:figure5} for an illustration.  To ensure
that~$B_1$ is an $O(n^3)$-time $uv$-trailblazer of degree $2$ for~$G$,
the rest of the phase proves that~$(a,t)$ is a trail marker for $B_1$
by showing that the iteration with $c=s$ reports a~$uv$-trail~$P_s$ of
$G$.

\begin{figure}
\centering
\includegraphics[width=.7\textwidth]{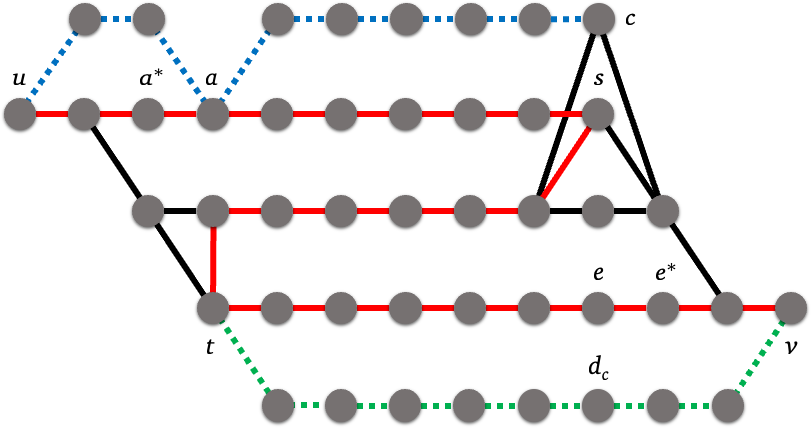}
\caption{An illustration for the proof that $B_1$ is a
  $uv$-trailblazer of degree two.  The red path denotes a shortest
  $uv$-trail $P$ of the $uv$-straight graph $G$.  The blue and green
  dotted paths denote a monotone~$uc$-path $S_c$ and a monotone
  $tv$-path $T_c$ of $G$ containing a precomputed pair of wings
  for~$(a,t,c,d_c)$ that need not coincide with $P$ except at $a$ and
  $t$.}
\label{figure:figure5}
\end{figure}

Let $a^*$ be the neighbor of $a$ in the monotone $P[u,a]$, implying
$h(a^*)=h(t)-1$.  Let~$e^*$ be the neighbor of $e$ in the
monotone~$P[e,v]$, implying $h(e^*)=h(s)+1$.  Since $P[a^*,s]$
and~$P[t,e^*]$ form a pair of wings for~$(a,t,s,e)$ in $G$, there is a
$d_s$ such that $(a,t,s,d_s)$ is winged in $G$.  Let~$(W_1,W_2)$ be
the pair of wings for~$(a,t,s,d_s)$ in $G$ obtained from $D$.  The
monotone~$us$-path~$S_s$ of $G$ containing~$W_1$ is a sidetrack
for~$P$, since the monotone~$tv$-path~$T_s$ of~$G$ containing $W_2$
satisfies Conditions~S\ref{S1} and~S\ref{S2} for $S_s$.  By
Lemma~\ref{lemma:lemma3}, each vertex~$x$ of~$P[s,t]$ satisfies
$h(t)\leq h(x)\leq h(s)$.  By Lemma~\ref{lemma:lemma4},~$S_s-s$
and~$P[s,t]-s$ are anticomplete in~$G$, implying that~$P[s,t]$ is a
path of $G_s$. Since~$s$ and $t$ are connected in $G_s$, the
subroutine call~$B_1(a,t)$ outputs a~$uv$-trail $P_s$ of $G$ in the
iteration with~$c=s$. Hence,~$(a,t)$ is indeed a trail marker of $B$.
One can verify that~$P_s$ is also a shortest $uv$-trail of the
$uv$-straight graph $G$, although $d_s$ need not be~$e$. Thus, we have
an~$O(n^5)$-time algorithm on an $n$-vertex general
(respectively,~$uv$-straight) graph $G$ that either reports a general
(respectively, shortest)~$uv$-trail of $G$ or ensures that $G$ is
$uv$-trailless.

\subparagraph*{Phase~2}
Since many prefixes of a long sidetrack for a shortest $uv$-trail $P$
of $G$ remain sidetracks for $P$, an edge can be deleted and then
inserted back $\Omega(n)$ times in Phase~1. Phase~2 avoids the
redundancy by processing the sidetracks in the decreasing order of
their lengths.  Let $B_2$ be the following subroutine that takes a
pair $(a,b)$ of $V(G)$ as the only argument.  Obtain in overall
$O(n^2)$ time from $D$ each set~$C_i$ with $0\leq i\leq h(v)$ that
consists of the vertices $c$ of~$G$ with~$h(c)=i$ such that $G$ admits
a winged quadruple~$(a,b,c,d_c)$ for some vertex $d_c$.  Let~$C$ be
the union of all $C_i$ with~$0\leq i\leq h(v)$.  Obtain in overall
$O(n^2)$ time from~$D$ for each vertex $c\in C$ (i) a
monotone~$uc$-path $S_c$ of $G$ containing~$a$ and (ii) a monotone
$bv$-path $T_c$ with
\[
d_{G[S_c\cup T_c]}(u,v)>\|S_c\|+\|T_c\|.
\]
Obtain the subgraph $H$ of $G$ induced by the vertices with heights at
least $h(a)$ in $O(n^2\log^2 n)$ time by the dynamic data structure of
Lemma~\ref{lemma:lemma7}.  Iteratively perform the following steps in
the decreasing order of the indices $i$ with $h(a)\leq i< h(v)$:
\begin{enumerate}
\item 
Delete from $H$ the incident edges of $N_G[S_c-c]\setminus \{c\}$ in
$G$ for all $c\in C_i$.

\item 
Insert to $H$ the incident edges of $C_i$ in $G$.

\item
Delete from $H$ all edges $xy$ of $G$ with $h(x)=i$ and $h(y)=i+1$.

\item
If $b$ is not connected to any $c\in C_i$ in $H$, then proceed to the
next iteration. Otherwise, let $c$ be an arbitrary vertex of $C_i$
that is connected to $b$ in $H$.  Exit the loop and report
the~$O(n^2)$-time obtainable concatenation $P_c$ of $S_c$ and a
shortest $cv$-path of $G[H\cup T_c]$.
\end{enumerate}
Since $S_c-c$ and $T_c-b$ are anticomplete in $G$ and the height of
each neighbor of $c$ in $H$ is at most $h(c)$, any arbitrary reported
$uv$-path $P_c$ of $G$ is a $uv$-trail of $G$.

Throughout all iterations, the incident edges of each vertex of $G$ is
deleted $O(1)$ times by Step~1, each edge of $G$ is updated $O(1)$
times by Steps~2 and~3, and each vertex $c\in C$ is queried $O(1)$
times by Step~4.  Thus, each subroutine call $B_2(a,b)$ runs in
$O(n^2\log^2 n)$ time.

Let~$P$ be an arbitrary shortest $uv$-trail of~$G$ with twist
pair~$(s,t)$.  As in Phase~1, let~$a$ (respectively,~$e$) be the
vertex of the monotone~$P[u,s]$ (respectively,~$P[t,v]$) with
$h(a)=h(t)$ (respectively,~$h(e)=h(s)$).  The rest of the phase proves
that~$(a,t)$ is a trail marker for $B_2$ by showing that an iteration
with $i\geq h(s)$ in the loop of the subroutine call~$B_2(a,t)$
reports a~$uv$-trail $P_c$ of $G$. See Figure~\ref{figure:figure6} for
an illustration.

\begin{figure}
\centering
\includegraphics[width=.7\textwidth]{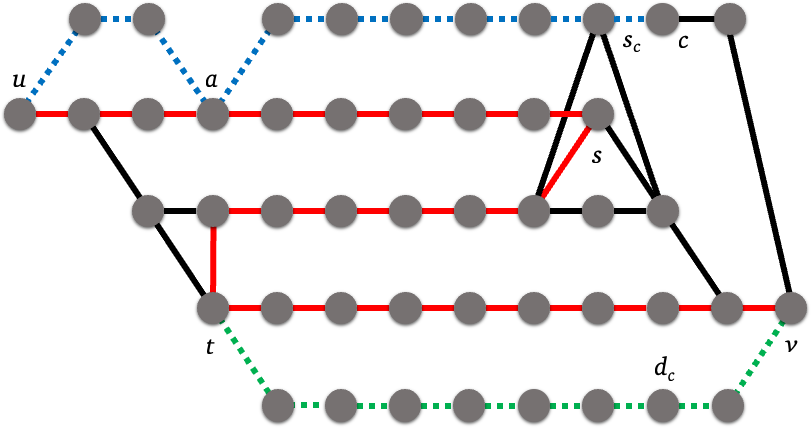}
\caption{An illustration for the proof that $B_2$ is a
  $uv$-trailblazer of degree two.  The red path denotes a shortest
  $uv$-trail $P$ of the $uv$-straight graph $G$.  The blue and green
  dotted paths denote a monotone~$uc$-path $S_c$ and a monotone
  $tv$-path $T_c$ of $G$ containing a precomputed pair of wings
  for~$(a,t,c,d_c)$ that need not coincide with $P$ except at $a$ and
  $t$. $S_c[u,s_c]$ remains a sidetrack for $P$.}
\label{figure:figure6}
\end{figure}

If an iteration of $B_2(a,t)$ with $i\geq h(s)+1$ reports a $uv$-trail
of $G$ (that need not be shortest), then we are done. Otherwise, we
show that the iteration with $i=h(s)$ has to report a~$uv$-trail of
$G$.  For each~$c\in C$ with $h(c)\geq i$, let $s_c$ be the unknown
vertex of $S_c$ with~$h(s_c)=i$. $S_c[u,s_c]$ remains a sidetrack for
$P$, since $T_c$ still satisfies Conditions~S\ref{S1} and~S\ref{S2}
for $S_c[u,s_c]$. Thus, $s_c\in C_i$.  By Lemma~\ref{lemma:lemma4},
$S_c[u,s_c]-s_c$ and~$P[s,t]-s$ are anticomplete in~$G$ even if
$S_c[u,s_c]$ need not be $S_{s_c}$.  As a result,
$P[s,t]-N_{P[s,t]}[s]$ is a path of the $H$ at the completion of
Step~1 in the $i$-th iteration.  By $s\in C_i$ and
Lemma~\ref{lemma:lemma3}, $P[s,t]$ is a path of the graph~$H$ at the
completion of Step~3 in the $i$-th iteration.  Therefore, $s$ is a
$c\in C_i$ that is connected to~$t$ in $H$.  Step~4 in the $i$-th
iteration has to output a (shortest) $uv$-trail~$P_c$ of $G$ for
some~$c\in C_i$ that need not be~$s$.  Thus, we have
an~$O(n^{4.75})$-time algorithm that either obtains a $uv$-trail of
$G$ or ensures that~$G$ is~$uv$-trailless. A reported~$uv$-trail of
$G$ by this~$O(n^{4.75})$-time algorithm need not be a shortest
$uv$-trail of~$G$, since we cannot afford to spend $O(n^2)$ time, as
in Phase~1, for each $c\in C$ that is connected to $t$ in the graph
$H$ at the~$h(c)$-th iteration to obtain a shortest~$cv$-path of
$G[H\cup T_c]$.

\section{Concluding remarks}
\label{section:section4}
We show an $O(n^{4.75})$-time algorithm for computing a $uv$-trail of
an $n$-vertex undirected unweighted graph $G$ with $\{u,v\}\subseteq
V(G)$.  The key to our improved algorithm is the observation regarding
an arbitrary shortest $uv$-trail of a $uv$-straight graph $G$
described by Lemma~\ref{lemma:lemma4}.  The inequality of
Lemma~\ref{lemma:lemma4} is stronger than the condition that~$S-c$
and~$P[s,t]-s$ are anticomplete in~$G$. As a matter of fact, the
latter suffices for our~$uv$-trailblazers in~\S\ref{section:section2}
and~\S\ref{section:section3}.  Thus, a further
improved~$uv$-trailblazer might be possible if the wings for a winged
quadruple can be obtained more efficiently.  As mentioned in Phase~1
of~\S\ref{section:section3}, a shortest $uv$-trail, if any, of a
$uv$-straight~$G$ can be obtained by our~$B_1$-based trailblazing
algorithm in~$O(n^5)$ time.  Detecting a $uv$-trail with length at
least $2d_G(u,v)$ is NP-complete~\cite[Theorem~1.6]{BergerSS21-dm}.
It is of interest to see if a shortest $uv$-trail or a $uv$-trail
having length at least $d_G(u,v)+k$ for a positive $k=O(1)$ in a
general $G$ can be obtained in polynomial time.  It is also of
interest to see whether the one-to-all (respectively, all-pairs)
version of the problem can be solved in time much lower
than~$O(n^{5.75})$ (respectively, $O(n^{6.75})$).

\section*{Acknowledgements}
Research of Hsueh-I Lu is supported by grants
110--2221--E--002--075--MY3 and 107--2221--E--002--032--MY3 from the
National Science and Technology Council (formerly known as Ministry of
Science and Technology).


\end{document}